\documentclass[aps, pra, reprint, superscriptaddress]{revtex4-2}

\pdfoutput=1

\usepackage{header}

\graphicspath{{graphics/}}

\begin{document}

\bibliographystyle{apsrev4-2}

\title{Efficient Algorithms for Weakly-Interacting Quantum Spin Systems}

\author{Ryan L. Mann}
\email{mail@ryanmann.org}
\homepage{http://www.ryanmann.org}
\affiliation{Centre for Quantum Computation and Communication Technology, Centre for Quantum Software and Information, School of Computer Science, Faculty of Engineering \& Information Technology, University of Technology Sydney, NSW 2007, Australia}

\author{Gabriel Waite}
\email{gabriel.waite@student.uts.edu.au}
\affiliation{Centre for Quantum Computation and Communication Technology, Centre for Quantum Software and Information, School of Computer Science, Faculty of Engineering \& Information Technology, University of Technology Sydney, NSW 2007, Australia}

\begin{abstract}
    We establish efficient algorithms for weakly-interacting quantum spin systems at arbitrary temperature. In particular, we obtain a fully polynomial-time approximation scheme for the partition function and an efficient approximate sampling scheme for the thermal distribution over a classical spin space. Our approach is based on the cluster expansion method and a standard reduction from approximate sampling to approximate counting.
\end{abstract}

\maketitle

\section{Introduction}
\label{section:Introduction}

The cluster expansion provides a powerful tool for developing approximation algorithms for statistical mechanical systems. This method has been successfully applied to obtain efficient algorithms for several models, including the hardcore model~\cite{helmuth2020algorithmic, jenssen2020algorithms, chen2019fast, cannon2020counting, jenssen2020independent, jenssen2023approximately, galvin2024zeroes, collares2025counting}, the Potts model~\cite{helmuth2020algorithmic, chen2019fast, borgs2020efficient, helmuth2023finite, carlson2024algorithms}, and quantum spin systems~\cite{mann2021efficient, helmuth2023efficient, mann2024algorithmic}. For quantum spin systems, efficient algorithms have been established at high temperature~\cite{mann2021efficient, mann2024algorithmic} and at low temperature for stable quantum perturbations of classical spin systems~\cite{mann2024algorithmic}.

In this paper, we apply this method to establish efficient algorithms for weakly-interacting quantum spin systems at arbitrary temperature. Our main results are a fully polynomial-time approximation scheme for the partition function and an efficient approximate sampling scheme for the thermal distribution over a classical spin space. Our algorithms are a natural extension of algorithms for high-temperature systems to weakly-interacting systems at arbitrary temperature.

Our approach is based on the algorithmic cluster expansion methods of Helmuth, Perkins, and Regts~\cite{helmuth2020algorithmic} and Borgs et al.~\cite{borgs2020efficient} for bounded-degree graphs. We apply their approach in the setting of bounded-degree bounded-rank multihypergraphs using the framework of Ref.~\cite{mann2024algorithmic}. That is, we formulate the partition function as an abstract polymer model following Koteck\'y and Preiss~\cite{kotecky1986cluster}. When the polymer weights satisfy a suitable decay condition, the cluster expansion provides a convergent power series representation for the logarithm of the partition function, which can be truncated to obtain efficient approximation algorithms. We obtain the sampling algorithm via a standard reduction from approximate sampling to approximate counting based on the chain rule for conditional probabilities. We note that our results extend naturally to fermionic systems; however, we focus on spin systems for clarity of presentation. We also note that these methods can be used to approximate expectation values of local observables, which can be obtained from derivatives of the partition function.

Our results complement classical algorithms for approximating the ground state energy and correlation functions of weakly-interacting quantum spin systems~\cite{bravyi2008polynomial} and efficient algorithms for approximating partition functions of weakly-interacting fermionic systems~\cite{chen2025convergence}. Our results further complement quantum algorithms for rapid mixing of thermal samplers and ground state preparation for weakly-interacting quantum systems~\cite{tong2025fast, smid2025polynomial, zhan2026rapid, smid2025rapid}. In particular, our results provide an efficient classical algorithm in the setting of Ref.~\cite{smid2025rapid}.

This paper is structured as follows. In Section~\ref{section:Preliminaries}, we introduce the necessary preliminaries. Then, in Section~\ref{section:ApproximateComputationOfThePartitionFunction}, we establish our algorithm for approximating the partition function. In Section~\ref{section:ApproximateSamplingFromTheThermalDistribution}, we establish our algorithm for approximately sampling from the thermal distribution over a classical spin space. Finally, we conclude in Section~\ref{section:ConclusionAndOutlook} with some remarks and open problems.

\section{Preliminaries}
\label{section:Preliminaries}

\subsection{Graph Theory}
\label{section:GraphTheory}

A \emph{multihypergraph} is a generalisation of a graph in which multiple edges between vertices and edges between any number of vertices are allowed. We consider multihypergraphs with uniquely labelled edges. Let $G=(V, E)$ be a multihypergraph. We denote the \emph{order} of $G$ by $\abs{G}\coloneqq\abs{V(G)}$ and the \emph{size} of $G$ by $\norm{G}\coloneqq\abs{E(G)}$. The \emph{maximum degree} of $G$ is the maximum degree over all vertices and the \emph{rank} of $G$ is the maximum cardinality over all edges.

\subsection{Quantum Spin Systems}
\label{section:QuantumSpinSystems}

A \emph{quantum spin system} is modelled by a multihypergraph $G=(V, E)$. At each vertex $v$ of $G$, there is a $d$-dimensional Hilbert space $\mathcal{H}_v$ with $d<\infty$. The Hilbert space on the multihypergraph is given by $\mathcal{H}_G\coloneqq\bigotimes_{v \in V}\mathcal{H}_v$. We consider Hamiltonians of the form $H_\Phi + \lambda H_\Psi$, where $H_\Phi$ is non-interacting, $H_\Psi$ is a local perturbation, and $\lambda\in\mathbb{C}$ is a parameter. The non-interacting Hamiltonian is defined by $H_\Phi\coloneqq\sum_{v \in V}\Phi_v$, where $\Phi$ assigns a self-adjoint operator $\Phi_v$ on $\mathcal{H}_v$ to each vertex $v$ of $G$. The local perturbation Hamiltonian is defined by $H_\Psi\coloneqq\sum_{e \in E}\Psi_e$, where $\Psi$ assigns a self-adjoint operator $\Psi_e$ on $\bigotimes_{v \in e}\mathcal{H}_v$ to each edge $e$ of $G$. At inverse temperature $\beta$, the \emph{partition function} $Z_G(\beta,\lambda)$ is defined by $Z_G(\beta,\lambda)\coloneqq\Tr\left[e^{-\beta(H_\Phi + \lambda H_\Psi)}\right]$, and the \emph{thermal state} $\rho_G(\beta,\lambda)$ is defined by $\rho_G(\beta,\lambda)\coloneqq\left(Z_G(\beta,\lambda)\right)^{-1}e^{-\beta(H_\Phi + \lambda H_\Psi)}$. The \emph{thermal distribution} $\mu_{\rho_G(\beta,\lambda)}$ over the classical spin space $[d]^V$ is defined by $\mu_{\rho_G(\beta,\lambda)}(x)\coloneqq\expval{\rho_G(\beta,\lambda)}{x}$ for all $x\in[d]^V$.

We shall restrict our attention to quantum spin systems modelled by bounded-degree bounded-rank multihypergraphs. We shall assume that the trace is normalised so that $\Tr(\mathbb{I})=1$, which is equivalent to rescaling the partition function by a multiplicative factor. Further, we shall assume that $\norm{\Phi_v}\leq1$ for every $v \in V$ and that $\norm{\Psi_e}\leq1$ for every $e \in E$, where $\norm{\;\cdot\;}$ denotes the operator norm. Note that this is always possible by a rescaling of $\beta$ and $\lambda$.

\subsection{Abstract Polymer Models}
\label{section:AbstractPolymerModels}

An \emph{abstract polymer model} is a triple $(\mathcal{C},w,\sim)$, where $\mathcal{C}$ is a countable set whose elements are called polymers, $w:\mathcal{C}\to\mathbb{C}$ is a function that assigns a \emph{weight} $w_\gamma\in\mathbb{C}$ to each polymer $\gamma\in\mathcal{C}$, and $\sim$ is a \emph{symmetric compatibility relation} such that each polymer is incompatible with itself. A set of polymers is called \emph{admissible} if all the polymers in the set are pairwise compatible. Note that the empty set is admissible. The \emph{abstract polymer partition function} $Z(\mathcal{C},w)$ is defined by
\begin{equation}
    Z(\mathcal{C},w) \coloneqq \sum_{\Gamma\in\mathcal{G}}\prod_{\gamma\in\Gamma}w_\gamma, \notag
\end{equation}
where the sum is over all admissible sets of polymers $\mathcal{G}$ from $\mathcal{C}$. We consider abstract polymer models in which the polymers are connected subgraphs of bounded-degree bounded-rank multihypergraphs and compatibility is defined by vertex disjointness. Accordingly, for a polymer $\gamma$, we denote its order by $\abs{\gamma}\coloneqq\abs{V(\gamma)}$ and its size by $\norm{\gamma}\coloneqq\norm{E(\gamma)}$.

\subsection{Abstract Cluster Expansion}
\label{section:AbstractClusterExpansion}

We now define the \emph{abstract cluster expansion}~\cite{kotecky1986cluster, friedli2017statistical}. Let $\Gamma$ be a non-empty ordered tuple of polymers. The \emph{incompatibility graph} $H_\Gamma$ is the graph whose vertex set is $\Gamma$ and has an edge between vertices $\gamma$ and $\gamma'$ if and only if they are incompatible. $\Gamma$ is called a \emph{cluster} if its incompatibility graph $H_\Gamma$ is connected. Let $\mathcal{G}_C$ denote the set of all clusters of polymers from $\mathcal{C}$. The \emph{abstract cluster expansion} is a formal power series for $\log(Z(\mathcal{C},w))$ in the variables $w_\gamma$, defined by
\begin{equation}
    \log(Z(\mathcal{C},w)) \coloneqq \sum_{\Gamma\in\mathcal{G}_C}\varphi(H_\Gamma)\prod_{\gamma\in\Gamma}w_\gamma, \notag
\end{equation}
where $\varphi(\;\cdot\;)$ denotes the \emph{Ursell function}, defined by
\begin{equation}
    \varphi(H) \coloneqq \frac{1}{\abs{H}!}\sum_{\substack{S \subseteq E(H) \\ \text{spanning} \\ \text{connected}}}(-1)^\abs{S}. \notag
\end{equation}
Our approach is based on representing the partition function of a quantum spin system as an abstract polymer model and applying the cluster expansion to approximate the logarithm of the partition function.

\subsection{Approximate Counting \& Sampling}
\label{section:ApproximateCountingSampling}

We now introduce some important concepts in approximate counting and sampling. A \emph{fully polynomial-time approximation scheme} for a sequence of complex numbers $(z_n)_{n\in\mathbb{N}}$ is a deterministic algorithm that, for any $n$ and $\epsilon>0$, outputs a complex number $\hat{z}_n$ such that $\abs{z_n-\hat{z}_n}\leq\epsilon\abs{z_n}$ in time polynomial in $n$ and $1/\epsilon$. An \emph{efficient approximate sampling scheme} for a sequence of probability distributions $(\mu_n)_{n\in\mathbb{N}}$ is a randomised algorithm that, for any $n$ and $\epsilon>0$, outputs a sample from a distribution $\hat{\mu}_n$ such that $\norm{\mu_n-\hat{\mu}_n}_\text{TV}\leq\epsilon$ in time polynomial in $n$ and $1/\epsilon$.

\section{Approximate Computation of the Partition Function}
\label{section:ApproximateComputationOfThePartitionFunction}

In this section we establish an efficient approximation scheme for the partition function of weakly-interacting quantum spin systems at arbitrary temperature. Our approach is based on the algorithmic approach to cluster expansions due to Helmuth, Perkins, and Regts~\cite{helmuth2020algorithmic} and Borgs et al.~\cite{borgs2020efficient} for bounded-degree graphs. We apply their approach in the setting of bounded-degree bounded-rank multihypergraphs using the framework of Ref.~\cite{mann2024algorithmic}. The resulting algorithm is similar to those for quantum partition functions at high temperature~\cite{mann2021efficient, mann2024algorithmic}. Our result is formalised in the following theorem.

\begin{theorem}
    \label{theorem:ApproximationAlgorithmPartitionFunction}
    Fix $\Delta,r\in\mathbb{Z}_{\geq2}$ and $\beta>0$. Let $G=(V, E)$ be a multihypergraph of maximum degree at most $\Delta$ and rank at most $r$, and let $\lambda$ be a complex number such that
    \begin{equation}
        \abs{\lambda} \leq \frac{e^{-2r\beta}}{e^4\beta\Delta\binom{r}{2}}. \notag
    \end{equation}
    Then the cluster expansion for $\log(Z_G(\beta,\lambda))$ converges absolutely, $Z_G(\beta,\lambda)\neq0$, and there is a fully polynomial-time approximation scheme for $Z_G(\beta,\lambda)$.
\end{theorem}

\begin{remark}
    Theorem~\ref{theorem:ApproximationAlgorithmPartitionFunction} extends naturally to fermionic systems where the non-interacting Hamiltonian is on-site.
\end{remark}

We prove Theorem~\ref{theorem:ApproximationAlgorithmPartitionFunction} by verifying that the conditions required to apply the abstract polymer model approximation theorem of Ref.~\cite{mann2024algorithmic} are satisfied. That is, we show that (1) the partition function $Z_G(\beta,\lambda)$ admits a suitable abstract polymer model representation, (2) the polymer weights satisfy the desired bound, and (3) the polymer weights can be computed in the desired time. This is achieved in the following three lemmas.

\begin{lemma}[{restate=[name=restatement]PartitionFunctionAbstractPolymerModel}]
    \label{lemma:PartitionFunctionAbstractPolymerModel}
    The partition function $Z_G(\beta,\lambda)$ admits the following abstract polymer model representation.
    \begin{equation}
        Z_G(\beta,\lambda) = Z_G(\beta,0)\sum_{\Gamma\in\mathcal{G}}\prod_{\gamma\in\Gamma}w_\gamma, \notag
    \end{equation}
    where
    \begin{equation}
        w_\gamma \coloneqq \!\!\sum_{T \subseteq E(\gamma)}\!\!(-1)^{\abs{E(\gamma){\setminus}T}}\frac{\Tr\left[e^{-\beta\left(\sum_{v \in V(\gamma)}\Phi_v+\lambda\sum_{e \in T}\Psi_e\right)}\right]}{Z_\gamma(\beta,0)}. \notag
    \end{equation}
\end{lemma}

We prove Lemma~\ref{lemma:PartitionFunctionAbstractPolymerModel} in Appendix~\ref{section:PartitionFunctionAbstractPolymerModel}.

\begin{lemma}[{restate=[name=restatement]PartitionFunctionWeightBound}]
    \label{lemma:PartitionFunctionWeightBound}
    Fix $\Delta,r\in\mathbb{Z}_{\geq2}$ and $\beta>0$. Let $G=(V, E)$ be a multihypergraph of maximum degree at most $\Delta$ and rank at most $r$, and let $\lambda$ be a complex number such that
    \begin{equation}
        \abs{\lambda} \leq \frac{e^{-2r\beta}}{e^4\beta\Delta\binom{r}{2}}. \notag
    \end{equation}
    Then, for all polymers $\gamma\in\mathcal{C}$, the weight $w_\gamma$ satisfies
    \begin{equation}
        \abs{w_\gamma} \leq \left(\frac{1}{e^3\Delta\binom{r}{2}}\right)^\norm{\gamma}. \notag
    \end{equation}
\end{lemma}

We prove Lemma~\ref{lemma:PartitionFunctionWeightBound} in Appendix~\ref{section:PartitionFunctionWeightBound}.

\begin{lemma}
    \label{lemma:PartitionFunctionWeightAlgorithm}
    The weight $w_\gamma$ of a polymer $\gamma$ can be computed in time $\exp(O(\norm{\gamma}))$.
\end{lemma}
\begin{proof}
    The sum is over all subsets $T$ of $E(\gamma)$, of which there are $2^\norm{\gamma}$. For each of these subsets $T$, the trace may be evaluated in time $\exp(O(\norm{\gamma}))$ by diagonalising the sum of interactions and the partition function may be evaluated in time $\norm{\gamma}^{O(1)}$ by a straightforward factorisation argument.
\end{proof}

We now prove Theorem~\ref{theorem:ApproximationAlgorithmPartitionFunction}.

\begin{proof}[Proof of Theorem~\ref*{theorem:ApproximationAlgorithmPartitionFunction}]
    By Lemma~\ref{lemma:PartitionFunctionAbstractPolymerModel}, the partition function $Z_G(\beta,\lambda)$ admits an abstract polymer model representation, where the polymers are connected subgraphs of a bounded-degree bounded-rank multihypergraph, and compatibility is defined by vertex disjointness. By Lemma~\ref{lemma:PartitionFunctionWeightBound} and Lemma~\ref{lemma:PartitionFunctionWeightAlgorithm}, for all polymers $\gamma\in\mathcal{C}$, the weight $w_\gamma$ satisfies
    \begin{equation}
        \abs{w_\gamma} \leq \left(\frac{1}{e^3\Delta\binom{r}{2}}\right)^\norm{\gamma}, \notag
    \end{equation}
    and can be computed in time $\exp(O(\norm{\gamma}))$. The proof then follows from Ref.~\cite[Theorem 3]{mann2024algorithmic}, which states that there is an efficient approximation algorithm for abstract polymer model partition functions that satisfy these conditions.
\end{proof}

\section{Approximate Sampling from the Thermal Distribution}
\label{section:ApproximateSamplingFromTheThermalDistribution}

In this section we establish an efficient approximate sampling scheme for the thermal distribution over a classical spin space for weakly-interacting quantum spin systems at arbitrary temperature. Sampling from this distribution enables a direct comparison with quantum thermal state preparation algorithms. Our approach is based on a standard reduction from approximate sampling to approximate counting of marginal probabilities. Our result is formalised in the following theorem.

\begin{theorem}
    \label{theorem:ApproximationAlgorithmThermalDistribution}
    Fix $\Delta,r\in\mathbb{Z}_{\geq2}$ and $\beta>0$. Let $G=(V, E)$ be a multihypergraph of maximum degree at most $\Delta$ and rank at most $r$, and let $\lambda$ be a real number such that
    \begin{equation}
        \abs{\lambda} \leq \frac{e^{-2r\beta}}{e^4\beta\Delta\binom{r}{2}}. \notag
    \end{equation}
    Then there is an efficient approximate sampling scheme for the thermal distribution $\mu_{\rho_G(\beta,\lambda)}$  over the classical spin space $[d]^V$.
\end{theorem}

\begin{remark}
    Theorem~\ref{theorem:ApproximationAlgorithmThermalDistribution} extends naturally to fermionic systems where the non-interacting Hamiltonian is on-site.
\end{remark}

To prove Theorem~\ref{theorem:ApproximationAlgorithmThermalDistribution}, we use the following lemma, which is standard in the theory of approximate counting and sampling. It shows that the existence of a fully polynomial-time approximation scheme for marginal probabilities implies an efficient approximate sampling scheme for the distribution. We note that more general results are known for self-reducible problems~\cite{jerrum1986random, sinclair1989approximate}.

\begin{lemma}[{restate=[name=restatement]ApproximateCountingApproximateSampling}]
    \label{lemma:ApproximateCountingApproximateSampling}
    Let $n\in\mathbb{N}$ and $d\in\mathbb{Z}_{\geq2}$ be integers, and let $\mu$ be a probability distribution over $[d]^n$. If there exists a fully polynomial-time approximation scheme for the marginal probabilities $\mu(x)$ for all $x\in\bigcup_{S\subseteq[n]}[d]^S$, then there is an efficient approximate sampling scheme for $\mu$.
\end{lemma}

We prove Lemma~\ref{lemma:ApproximateCountingApproximateSampling} in Appendix~\ref{section:ApproximateCountingApproximateSampling}. We now prove Theorem~\ref{theorem:ApproximationAlgorithmThermalDistribution}.

\begin{proof}[Proof of Theorem~\ref*{theorem:ApproximationAlgorithmThermalDistribution}]
    By Theorem~\ref{theorem:ApproximationAlgorithmPartitionFunction}, there is a fully polynomial-time approximation scheme for $Z_G(\beta,\lambda)$, which extends to the marginal probabilities $\mu_{\rho_G(\beta,\lambda)}(x)$ for all $x\in\bigcup_{S \subseteq V}[d]^S$ by restricting the trace to basis states consistent with $x$, that is, summing only over basis states that agree with $x$ on $S$. The result then follows from Lemma~\ref{lemma:ApproximateCountingApproximateSampling}.
\end{proof}

We note that this approach can be applied with the results of Refs.~\cite{mann2021efficient, mann2024algorithmic} to obtain an efficient approximate sampling scheme for quantum spin systems at high temperature. A similar result in this setting was obtained in Ref.~\cite{yin2023polynomial} using related methods. Our approach leads to sharper bounds on the inverse temperature and demonstrates that such results can be obtained straightforwardly from results on approximate counting for quantum spin systems. At higher temperatures, efficient sampling algorithms have been developed with runtime polynomial in $\abs{V}$ and $\log(1/\epsilon)$~\cite{bakshi2024high, ramkumar2025high}.

\section{Conclusion \& Outlook}
\label{section:ConclusionAndOutlook}

We have established efficient algorithms for weakly-interacting quantum spin systems at arbitrary temperature. We obtained a fully polynomial-time approximation scheme for the partition function and an efficient approximate sampling scheme for the thermal distribution over a classical spin space.

It would be interesting to extend these methods to local perturbations of free-fermionic Hamiltonians, where efficient algorithms are known via the cumulant expansion~\cite{chen2025convergence}. It would also be interesting to obtain efficient algorithms with runtime polynomial in $\abs{V}$ and $\log(1/\epsilon)$, extending recent progress on high-temperature sampling algorithms~\cite{bakshi2024high, ramkumar2025high}. We note that efficient algorithms with a bound on the perturbation parameter that is independent of temperature can be obtained under additional assumptions by applying the framework of Ref.~\cite{helmuth2023efficient}.

\section*{Acknowledgements}

We thank Samuel Scalet for helpful discussions. RLM was supported by the ARC Centre of Excellence for Quantum Computation and Communication Technology (CQC2T), project number CE170100012. GW was supported by a scholarship from the Sydney Quantum Academy and supported by the ARC Centre of Excellence for Quantum Computation and Communication Technology (CQC2T), project number CE170100012.

\appendix

\onecolumngrid

\section{Proof of Lemma~\ref*{lemma:PartitionFunctionAbstractPolymerModel}}
\label{section:PartitionFunctionAbstractPolymerModel}

\PartitionFunctionAbstractPolymerModel*

\begin{proof}
    By the principle of inclusion-exclusion (see for example \cite[Theorem 12.1]{graham1995handbook}),
    \begin{align}
        Z_G(\beta,\lambda) &= \Tr\left[e^{-\beta(H_\Phi + \lambda H_\Psi)}\right] \notag \\
        &= \sum_{S \subseteq E}(-1)^\abs{S}\sum_{T \subseteq S}(-1)^\abs{T}\Tr\left[e^{-\beta\left(H_\Phi+\lambda\sum_{e \in T}\Psi_e\right)}\right]. \notag
    \end{align}
    We now extract an overall factor of $Z_G(\beta,0)=\Tr\left[e^{-\beta H_\Phi}\right]$. This gives
    \begin{equation}
        Z_G(\beta,\lambda) = Z_G(\beta,0)\sum_{S \subseteq E}(-1)^\abs{S}\sum_{T \subseteq S}(-1)^\abs{T}\frac{\Tr\left[e^{-\beta\left(\sum_{v\in\operatorname{supp}(S)}\Phi_v+\lambda\sum_{e \in T}\Psi_e\right)}\right]}{\Tr\left[e^{-\beta\sum_{v\in\operatorname{supp}(S)}\Phi_v}\right]}. \notag
    \end{equation}
    For a subset $S \subseteq E$, let $\Gamma_S$ denote the set of maximally connected components of $S$. By factorising over these components, we have
    \begin{align}
        Z_G(\beta,\lambda) &= Z_G(\beta,0)\sum_{S \subseteq E}\prod_{\gamma\in\Gamma_S}\sum_{T \subseteq E(\gamma)}(-1)^{\abs{E(\gamma){\setminus}T}}\frac{\Tr\left[e^{-\beta\left(\sum_{v \in V(\gamma)}\Phi_v+\lambda\sum_{e \in T}\Psi_e\right)}\right]}{Z_\gamma(\beta,0)} \notag \\
        &= Z_G(\beta,0)\sum_{S \subseteq E}\prod_{\gamma\in\Gamma_S}w_\gamma \notag \\
        &= Z_G(\beta,0)\sum_{\Gamma\in\mathcal{G}}\prod_{\gamma\in\Gamma}w_\gamma. \notag
    \end{align}
    This completes the proof.
\end{proof}

\section{Proof of Lemma~\ref*{lemma:PartitionFunctionWeightBound}}
\label{section:PartitionFunctionWeightBound}

\PartitionFunctionWeightBound*

\begin{proof}
    Fix a polymer $\gamma$. Let $P$ denote the set of all sequences of edges in $\gamma$. By applying the Duhamel expansion, it follows from the triangle inequality and the submultiplicativity of the norm that
    \begin{align}
        \abs{w_\gamma} &\leq \norm{\sum_{T \subseteq E(\gamma)}(-1)^\abs{T}\frac{\Tr\left[e^{-\beta\left(\sum_{v \in V(\gamma)}\Phi_v+\lambda\sum_{e \in T}\Psi_e\right)}\right]}{Z_\gamma(\beta,0)}} \notag \\
        &\leq \frac{e^{\beta\abs{\gamma}}}{Z_\gamma(\beta,0)}\sum_{\rho \in P}\frac{(\beta\abs{\lambda})^\abs{\rho}}{\abs{\rho}!}\prod_{e\in\rho}\norm{\Psi_e} \notag \\
        &\leq e^{2\beta\abs{\gamma}}\sum_{\rho \in P}\frac{(\beta\abs{\lambda})^\abs{\rho}}{\abs{\rho}!}. \notag
    \end{align}
    There are exactly $\genfrac{\{}{\}}{0pt}{}{n}{\norm{\gamma}}\norm{\gamma}!$ sequences $\rho$ of length $n$ whose support is $\gamma$, where $\genfrac{\{}{\}}{0pt}{}{n}{\norm{\gamma}}$ denotes the Stirling number of the second kind. Hence, we may write
    \begin{equation}
        \abs{w_\gamma} \leq e^{2\beta\abs{\gamma}}\sum_{n=\norm{\gamma}}^\infty\genfrac{\{}{\}}{0pt}{}{n}{\norm{\gamma}}\frac{\norm{\gamma}!}{n!}(\beta\abs{\lambda})^\abs{\rho} = e^{2\beta\abs{\gamma}}\left(e^{\beta\abs{\lambda}}-1\right)^\norm{\gamma}, \notag
    \end{equation}
    where we have used the identity $\sum_{n=k}^\infty\genfrac{\{}{\}}{0pt}{}{n}{k}\frac{x^n}{n!}=\frac{(e^x-1)^k}{k!}$. Now, since $\abs{\gamma}$ is at most $(r-1)\norm{\gamma}+1$, we have
    \begin{equation}
        \abs{w_\gamma} \leq \left(e^{2r\beta}\left(e^{\beta\abs{\lambda}}-1\right)\right)^\norm{\gamma}. \notag
    \end{equation}
    By taking $\abs{\lambda}\leq\left(e^4\beta\Delta\binom{r}{2}e^{2r\beta}\right)^{-1}$, we obtain
    \begin{equation}
        \abs{w_\gamma} \leq \left(\frac{1}{e^3\Delta\binom{r}{2}}\right)^\norm{\gamma}, \notag
    \end{equation}
    completing the proof.
\end{proof}

\section{Proof of Lemma~\ref*{lemma:ApproximateCountingApproximateSampling}}
\label{section:ApproximateCountingApproximateSampling}

\ApproximateCountingApproximateSampling*

\begin{proof}
    We construct an efficient approximate sampling scheme by applying the chain rule for conditional probabilities. For a sequence $x=(x_i)_{i=1}^n$, we sample each $x_i\in[d]$ conditioned on the previously sampled values $x_{\prec i}\coloneqq(x_j)_{j=1}^{i-1}$, using the conditional probability
    \begin{equation}
        \hat{\mu}(x_i \mid x_{\prec i}) \coloneqq \frac{\hat{\mu}(x_{\prec i}x_i)}{\sum_{y\in[d]}\hat{\mu}(x_{\prec i}y)}, \notag
    \end{equation}
    where $\hat{\mu}(x_{\prec i}x_i)$ is an approximation to $\mu(x_{\prec i}x_i)$. These approximations are obtained via the assumed fully polynomial-time approximation scheme to within an error of $\epsilon/(3n)$. Then, for $0<\epsilon\leq1$,
    \begin{equation}
        \abs{\mu(x_i \mid x_{\prec i})-\hat{\mu}(x_i \mid x_{\prec i})} \leq \frac{\epsilon}{n}\mu(x_i \mid x_{\prec i}). \notag
    \end{equation}
    Hence, we sample $x\in[d]^n$ with probability 
    \begin{equation}
        \hat{\mu}(x) = \prod_{i=1}^n\hat{\mu}(x_i \mid x_{\prec i}), \notag
    \end{equation}
    such that
    \begin{equation}
        \abs{\mu(x)-\hat{\mu}(x)} \leq 2\epsilon\mu(x). \notag
    \end{equation}
    Therefore, the total variation distance between $\mu$ and $\hat{\mu}$ is bounded by
    \begin{equation}
        \norm{\mu-\hat{\mu}}_\text{TV} = \frac{1}{2}\sum_{x\in[d]^n}\abs{\mu(x)-\hat{\mu}(x)} \leq \epsilon\sum_{x\in[d]^n}\mu(x) = \epsilon. \notag
    \end{equation}
    This establishes an efficient approximate sampling scheme for $\mu$, completing the proof.
\end{proof}

\twocolumngrid

\bibliography{bibliography}

\end{document}